\documentclass[12pt]{article}
\usepackage{hyperref}

\usepackage{marvosym,verbatim,stackengine}
\usepackage[backrefs,alphabetic]{amsrefs}
\usepackage[margin=2cm]{geometry}
\usepackage{amsmath,amstext,amssymb,amsthm,bbm,xcolor}
\usepackage{graphicx,tikz}

\newtheorem{thm}{Theorem}
\newtheorem{prop}{Proposition}[section]
\newtheorem{lemma}[prop]{Lemma}

\theoremstyle{remark}
\newtheorem{df}[prop]{Definition}
\newtheorem{rmk}[prop]{Remark}
\numberwithin{equation}{section}

\def\TT{\mathbb T}

\begin{document}

\title{Anderson localisation in stationary ensembles of quasiperiodic operators}
\author{Victor Chulaevsky\textsuperscript1 and Sasha Sodin\textsuperscript2}
\maketitle
\footnotetext[1]{D\'{e}partement de Math\'{e}matiques,
Universit\'{e} de Reims, Moulin de la Housse, B.P. 1039,
51687 Reims Cedex 2, France. Email:
victor.tchoulaevski@univ-reims.fr.}
\footnotetext[2]{School of Mathematical Sciences, Queen Mary University of London, London E1 4NS, United Kingdom. Email:
a.sodin@qmul.ac.uk.}

\begin{abstract}An ensemble of quasi-periodic discrete   Schr\"{o}dinger operators with an arbitrary number
of basic frequencies is considered, in a lattice of arbitrary dimension, in which the hull function
is a realisation of a stationary Gaussian process on the torus. We show that, for almost
every element of the ensemble, the quasi-periodic operator boasts Anderson localization
with simple pure point spectrum at strong coupling. One of the ingredients of the proof is a new lower bound on the interpolation error for stationary Gaussian processes on the torus (also known as local non-determinism).
\end{abstract}

{\centering\em \qquad\qquad Dedicated to Ya.\ G.\ Sinai on his 85th birthday\par}

\section{Introduction}
We consider quasiperiodic Schr\"odinger operators on $\mathbb Z^d$ (equipped with the graph metric $\| \cdot \|$), for arbitrary $d \geq 1$ and  an arbitrary number of frequencies $\nu \geq 1$. Let $\TT^\nu = (\mathbb R/ \mathbb Z)^\nu$; fix a continuous function $v: \mathbb T^\nu \to \mathbb R$, a $\nu \times d$ frequency matrix $\alpha = (\alpha_{ij})$, an initial point $\omega \in \TT^\nu$, and a coupling $g > 0$, and define an operator $H = H(\omega; g)$ on $\ell_2(\mathbb Z^d)$ by
\begin{equation}\label{eq:schr} (H(\omega; g) f)(x) = \sum_{\|y-x\|=1} f(y) + g v(\omega + \alpha x) f(x)~,.\end{equation}
Operators of the form $H(\omega;g)$ form an important subclass of metrically transitive (ergodic) operators \cite{FiPbook}.

Operators of the form (\ref{eq:schr}) have been intensively studied for $d = \nu = 1$. It was found that for large $g \geq g_0$ and Diophantine $\alpha$, the operator exhibits Anderson localisation, manifesting itself in pure point spectrum with exponentially decaying eigenfunctions. This phenomenon has been rigorously established first for the Maryland model $v(\omega) = \tan (2 \pi \omega)$ and for more general tangent-like potentials \cite{FiP84,Sim85,BLS83,JK} (following the physical work \cite{FGP84}), then for the Almost Mathieu model $v(\omega) = \cos(2\pi \omega)$ and more general cosine-like potentials \cite{Sin87,FSW90,J0,J1}, and, more recently, for general analytic potentials \cite{BG00,Bourg} and further for potentials in Gevrey classes \cite{Kl1,Kl2}. We refer to the survey \cite{JM} for a review of the state of art. In \cite{BGS01}, Anderson localisation was established for a class of analytic potentials for $d = 1$ and $\nu = 1,2$.

Much less is known for $d > 1$. The analysis of tangent-like potentials was extended to higher dimension in \cite{BLS83}. In \cite{Craig}, quasiperiodic potentials exhibiting pure point spectrum were constructed using an inverse spectral procedure. In \cite{BGS02}, Anderson localisation at strong coupling was proved for analytic potentials and $d = \nu = 2$; this result is perturbative, meaning that for each $\omega$ localisation holds outside a set of frequencies the measure of which tends to zero as $g \to \infty$. In \cite{Bourg07}, the result of \cite{BGS02} was extended to arbitrary $d = \nu$, and in \cite{JLS} -- to arbitrary $d$ and $\nu$. We also mention the work  \cite{KS} on delocalisation, i.e.
\ the existence of absolutely continuous spectrum, at weak coupling (for an operator in the continuum).

These results raised the question whether Anderson localisation persists when $v$ is less smooth, e.g.\ has a finite number of derivatives. Another question is whether localisation holds in the non-perturbative setting for $d > 1$, under a usual Diophantine condition on the frequency. As these questions are yet to be answered for explicit $v$ such as $v(\theta) = \sum_j \cos \theta_j$, it was suggested in \cite{C11,C14} to study the properties of (\ref{eq:schr}) for typical hull functions $v$: namely, $v$ is chosen as a realisation of a stochastic process on $\mathbb T^\nu$. Related ideas appeared in the work  \cite{Chan07}. In these works, Anderson localisation was established for $v$ sampled from a class of (non-stationary) stochastic processes, constructed to ensure the required properties. Here, we extend these results to the more natural class of stationary Gaussian processes on the torus:
\begin{equation}\label{eq:strg} v(\omega) = \sum_{\ell \in \mathbb (2\pi \mathbb Z)^\nu} \frac{g_l \cos \langle \omega, \ell \rangle + h_\ell \sin \langle \omega, \ell\rangle}{\sqrt{W(\ell)}}~, \quad \omega \in \mathbb T^\nu~,\end{equation}
where $g_\ell$ and $h_\ell$ are jointly independent standard Gaussian random variables, and $W: 2 \pi \mathbb Z^\nu \to \mathbb R_+$ is a spectral weight. Denote the underlying probability space by $(\Theta, \mathcal B^\Theta, \mathbb P^\Theta)$; to emphasise the dependence on $\theta$, we write $v(\omega) = v(\omega, \theta)$. Denote the operator corresponding to $\theta \in \Theta$ by $H(\omega,\theta; g)$.

\begin{thm}\label{t} Assume that $W: 2 \pi \mathbb Z^\nu \to \mathbb R_+$ is such that
\[ c \|\ell\|^{\nu + \delta} \leq  W(\ell) \leq C e^{C\|\ell\|^\zeta}~, \quad \ell \in 2\pi\mathbb Z^\nu~, \]
for some $\kappa, \zeta, \delta > 0$, and $C, c >0$, and that $\alpha$ satisfies the Diophantine condition
\begin{equation}\label{eq:dioph} \operatorname{dist} (\alpha x, \mathbb Z^\nu)  \geq c' \|x\|^{-A}~, \quad x \in \mathbb Z^d \setminus \{0\}\end{equation}
with some $A>0$ and $c' > 0$. If $(A+1)\zeta < 1$, then there exists a map $\Theta^+: \mathbb R_+ \to \mathcal B^\Theta$ such that $\mathbb P^\Theta(\Theta^+(g)) \to 1$ as $g \to +\infty$, and for every $\theta \in \Theta^+(g)$ and almost every $\omega \in \mathbb T^\nu$, the spectrum of the operator $H(\omega,\theta; g)$ constructed from (\ref{eq:strg})
is pure point, and every eigenfunction $\psi$ of $H(\omega, \theta; g)$ satisfies
\begin{equation}\label{eq:expdec}
 \sup_{x \in \mathbb Z^d} |\psi(x)| e^{\|x\|} < \infty~. \end{equation}
\end{thm}

\begin{rmk} According to a theorem of Groshev \cite{Grosh, BerVel}, for $\alpha$ in a set of full measure the condition (\ref{eq:dioph}) holds with any $A > d/\nu$.\end{rmk}

\begin{rmk} As part of the proof, we show in Lemma~\ref{l:weg} that the number of ``resonances'' is uniformly bounded. For processes with uniformly Lipschitz realisation, our uniform bound $k_{\max}= \nu + 1$ is optimal, as $\nu+1$-fold resonances are known to be topologically unavoidable. For a different class of Gaussian processes, the same conclusion was established in \cite{C11}.
\end{rmk}

The main theorem follows from two propositions. The first one, Proposition~\ref{t:1}, establishes the conclusion of Theorem~\ref{t} in a more abstract setting, when  $\omega + \alpha x$ in (\ref{eq:schr}) is replaced with an orbit of an ergodic action of $\mathbb Z^d$ on a metric probability space $\Omega$. The second one, Proposition~\ref{p:var}, confirms that the assumptions are satisfied for the process (\ref{eq:strg}).

\paragraph{A general localisation theorem}

In this section, we replace the torus $\mathbb T^{\nu}$ with a metric probability space $(\Omega, \mathcal B^\Omega, \mathbb P^\Omega, \operatorname{dist})$ of  finite metric dimension, i.e.\ we assume that there exists $\nu > 0$ (not necessarily integer) such that, for any $\epsilon \in (0, 1]$, $\Omega$ admits an $\epsilon$-net of cardinality at most $(C/\epsilon)^\nu$. Let $T: \Omega \times \mathbb Z^d \to \Omega$ be an ergodic action of $\mathbb{Z}^d$ on $\Omega$ satisfying the Diophantine property
\begin{equation}\label{eq:upa} \text{\bf (UPA)}_A \quad \inf_\omega \min_{0 < \|x\| \leq L} \operatorname{dist}(T^x \omega, \omega) \geq c L^{-A}~, \quad L \in \mathbb N~. \end{equation}
For the case of $\mathbb T^\nu$ with the action $T^x \omega = \omega + \alpha x$, the condition $\text{\bf (UPA)}_A$ boils down to the Diophantine property (\ref{eq:dioph}).

Let $(\Theta, \mathbb B^\Theta, \mathbb P^\Theta)$ be an additional probability space, and let $v(\omega, \theta)$ be a (modification of a) stochastic process defined on $\Theta$ and taking values in the space of uniformly $\kappa$-H\"older-continuous functions from $\Omega$ to $\mathbb R$ (for some fixed $\kappa>0$), so that for any $\omega \in \Omega$ the conditional distribution of the random variable $v(\omega, \cdot)$ conditioned on the values of the process in the complement to the $\epsilon$-neighbourhood $Q_\epsilon(\omega)$ of $\omega$ is absolutely continuous and admits a density satisfying the local interpolation bound
\begin{equation}\label{eq:lib} \text{\bf (LIB)}_\eta \quad
p_\omega (t \, \mid \, \Omega \setminus Q_\epsilon(\omega)) \leq \exp(C \epsilon^{-\eta})~, \quad \epsilon \in (0,\epsilon_0]!.\end{equation}
Then we replace (\ref{eq:schr}) with the more general metrically transitive operator
\begin{equation}\label{eq:schr''} (H(\omega, \theta; g) f)(x) = \sum_{\|y-x\|=1} f(y) + g v(T^x \omega, \theta) f(x)~.\end{equation}

\begin{prop}\label{t:1} Assume that the assumptions $\text{\bf (UPA)}_A$ and $\text{\bf (LIB)}_\eta$ hold with $A$ and $\eta$ such that $A\eta < 1$. Then there exists a map $\Theta^+: \mathbb R_+ \to \mathcal B^\Theta$ such that $\mathbb P^\Theta(\Theta^+(g)) \to 1$ as $g \to +\infty$, and for every $\theta \in \Theta^+(g)$ and almost every $\omega \in \Omega$, the spectrum of the operator $H(\omega,\theta; g)$
is pure point, and every eigenfunction $\psi$ satisfies
\begin{equation}\label{eq:expdec'}
 \sup_x |\psi(x)|   e^{\|x\|} < \infty~. \end{equation}
\end{prop}

\begin{rmk}  Proposition~\ref{t:1} (and, accordingly, also Theorem~\ref{t}) can be strengthened in several directions, without invoking new methods:
\begin{enumerate}
\item the rate of exponential decay (\ref{eq:expdec}) can be improved to $\sup_x |\psi(x)| e^{m_g\|x\|} < \infty$ for an arbitrary $m_g = o(g)$;
\item on the event $\Theta^+(g)$, the operator can be shown to exhibit dynamical localisation (our bounds on the eigenfunctions are sufficient to control the eigenfunction correlators \cite{Ai94,ASFH01,AWbook});
\item on the event $\Theta^+(g)$, the spectrum of $H$ can be shown to be simple (see \cite{C14}, building on the method of \cite{KM06}).
\end{enumerate}
\end{rmk}

\paragraph{Interpolation of stationary processes}

Consider a stationary Gaussian process
\begin{equation}\label{eq:strg'} v(\omega) = \sum_{\ell \in \mathbb (2\pi \mathbb Z)^\nu} \frac{g_\ell \cos \langle \omega, \ell \rangle + h_\ell \sin \langle \omega, \ell\rangle}{\sqrt{W(\ell)}}~, \quad \omega \in \TT^\nu~,\end{equation}
as in (\ref{eq:strg}). For $0 < \epsilon \leq 1/2$ let
\[ \mathbf V(\epsilon) = \operatorname{Var} \left( v(\omega) \, \big| \, \{ v(\omega') \, : \, \omega' \in \TT^\nu~, \,\|\omega' - \omega\| \geq \epsilon \}\right) \]
be the conditional variance of $v(\omega)$ conditioned on the complement to the $\epsilon$-neighbourhood of $\omega$ (here and forth $\| \cdot \| = \| \cdot \|_\infty$ is the $\ell_\infty$ distance from $0$ on $\TT^\nu$).

\begin{prop}\label{p:var} Assume that there exists a  non-decreasing function   $M:\mathbb R_+ \to \mathbb R_+$ such that
\begin{equation}\label{eq:exM} \int_{t_0}^\infty \frac{\log M(t)}{t^2} dt < \infty~, \quad K = \sum_{\ell \in 2\pi \mathbb Z^\nu} \frac{W(\ell)}{M(\|\ell\|)}  < \infty~.\end{equation}
Then for
\[ 0 < \epsilon \leq \min\left(\frac12, \frac{e}{2} \int_0^\infty \frac{\log M(t)}{t^2} dt \right)\]
the conditional variance $\mathbf V(\epsilon)$ admits the lower bound
\[ \mathbf V(\epsilon)  \geq  \frac{1}{C_\nu K \epsilon^{2\nu} \, M(S^{-1}(\frac2e \epsilon))}~,
\quad \text{where} \quad S(t) = \int_t^\infty \frac{\log M(\tau)}{\tau^2} d\tau~, \quad C_\nu = e^2 2^\nu~.\]
\end{prop}

\begin{rmk} The asymptotic behaviour of $\mathbf V(\epsilon)$ as $\epsilon \to + 0$ is an aspect of the interpolation problem for stationary Gaussian processes, going back to  \cite{Kolm}. The interpolation problem was studied, for the  $\nu = 1$ case of the full-space process
\begin{equation}\label{eq:contproc}\tilde v (\xi) = \int_{\mathbb{R^\nu}} \frac{\cos \langle \xi, \lambda\rangle dB_1(\lambda) + \sin \langle \xi, \lambda\rangle dB_2(\lambda)}{\sqrt{(2\pi)^d W(\lambda)}}~, \quad \xi \in \mathbb R^\nu~,\end{equation}
in \cite[\S 4.13 and Ch.\ 6]{DM} (where $B_1$ and $B_2$ are Brownian motions). The connection with the theory of de Branges spaces and Krein strings, established in these works, allows, in particular, to compute $\mathbf V(\epsilon)$ explicitly in several examples. A condition of the form (\ref{eq:exM}) is unavoidable: for sufficiently regular weights $W$, it holds for an appropriately chosen majorant $M$ whenever $V(\epsilon) \not\equiv 0$.

Quantitative bounds for $\mathbf V(\epsilon)$ in the $\nu=1$ case of (\ref{eq:contproc}) were obtained by  \cite{CuzickDupreez}, building on the work   \cite{Cuzick}. When applied to (\ref{eq:contproc}), our method yields marginally weaker bounds for $W(\lambda) \propto |\lambda|^\alpha$ and marginally stronger ones for any faster-growing $W$, particularly, for $W(\lambda) \propto \exp(\|\lambda\|^\zeta)$. Another advantage is that our estimate is somewhat more explicit, and adjusts easily to the process on the torus $\TT^\nu$ (for arbitrary $\nu$), as is required here. On the other hand, it is conceivable that a bound sufficient for Theorem~\ref{t} can be also obtained by the method of \cite{CuzickDupreez}.\end{rmk}

\paragraph{Proof of Theorem~\ref{t}.}
Assume that
\[ c \|\ell\|^{\nu + \delta} \leq W(\ell) \leq C \exp(C \|\ell\|^\zeta)~. \]
Fix $0 < \kappa < \delta$; the lower bound ensures that the realisations of $v$ are almost surely uniformly $\kappa$-H\"older continuous. From the upper bound,
\[ \sum_\ell \frac{W(\ell)}{M(\|\ell\|)} < \infty~, \quad \text{where} \quad M(t) = e^{2 C t^\zeta}~.\]
We apply Proposition~\ref{p:var}:
\[ S(t) = \int_t^\infty \frac{2C \tau^\zeta}{\tau^2} d\tau \leq C_1 t^{-(1-\zeta)}~, \quad
S^{-1} (\epsilon) \leq C_2 \epsilon^{-\frac{1}{1-\zeta}}~, \]
therefore
\[ \mathbf V(\epsilon) \geq \frac{1}{C_3 \exp(C_4 \epsilon^{-\frac{\zeta}{1-\zeta}})}~,\]
i.e.\ $\text{\bf (LIB)}_\eta$ holds with $\eta = \zeta / (1-\zeta)$.
The assumption $\zeta ( A+1) < 1$ ensures that $\eta A  < 1$, hence we can apply Proposition~\ref{t:1}.
\qed

\section{Multiscale analysis: Proof of Proposition~\ref{t:1}}\label{S:mult}
The proof of Proposition~\ref{t:1} is based on multi-scale analysis, originating in the work \cite{FrSp1} on random operators. Our version of the argument, building on \cite{C11,C14}, is organised as follows: a deterministic inductive procedure is established in Proposition~\ref{p:msa} of Section~\ref{s:ind}, and then, in Section~\ref{s:pft1}, we verify that the conditions of Proposition~\ref{p:msa} are satisfied for our random operator (on an event of full probability). The main technical difference compared to the works \cite{C11,C14} is the use of $2L \times L$ rectangles (and more generally $2L \times L \times \cdots \times L$ cuboids) instead of squares and cubes in the induction.

\subsection{Scale induction}\label{s:ind}

In this section, $H$ is a fixed discrete Schr\"odinger operator acting on $\ell_2(\mathbb Z^d)$. For a finite $B \subset \mathbb Z^d$, denote by $H_B$ the restriction of $H$ to $B$, i.e.\ $H_B = P_B H P_B^*$, where $P_B: \ell_2(\mathbb Z^d) \to \ell_2(B)$ is the coordinate projection. For $E \in \mathbb R$, let $G_E[H_B] = (H_B - E)^{-1}$ be the resolvent of $H_B$ at $E$.

The multi-scale induction involves the parameters $m > 0$, $b \in (0,1)$, $\gamma \in (2 - b, \infty)$ and $J \in \mathbb N$, which will be fixed throughout the argument (that is, one may  choose them tailored to the operator $H$). Their r\^oles are as follows:
\begin{itemize}
\item $m$ is a ``mass'', controlling the rate of exponential decay of the Green function in infinite volume;
\item $b$ is responsible for the deterioration of the mass: on the scale $L$, the mass will be $m(1 + L^{-(1-b)})$;
\item  $\gamma$ is responsible for the growth of scales: we fix $L_0$ (the scale of the box used as the induction base), and let $L_{k+1} = \lfloor L_k^\gamma \rfloor$;
\item $J \geq 1$ controls the number of ``resonances''.
\end{itemize}

\begin{df} A box is a product of $d$ intervals: $B = I_1 \times \cdots \times I_d \subset \mathbb Z^d$.  We denote by $\mathfrak B$ the collection of all boxes, and by $\mathfrak B_2$ the collection of sets $b_1 \setminus b_2$, where $b_1, b_2$ are boxes.

A box $R \subset \mathbb{Z}^d$ is called an $L$-rectangle if $d-1$ of the intervals in the product are of cardinality $2L+1$ (i.e.\ of length $2L$) and one is of cardinality $L+1$ (i.e.\ of length $L$).

The boundary of $s \subset \mathbb Z^d$ is the set $\partial s \subset \mathbb Z^d \times \mathbb Z^d$ of pairs $(u, u') \in s \times (\mathbb Z^d \setminus s)$ such that $\|u - u'\|=1$. The projection of $\partial s$ onto the first coordinate is denoted $\partial_{\text{in}} s (\subset s)$.
\end{df}

\begin{df} Given $E \in \mathbb R$, an $L$-rectangle $R$ is called $E$-regular if
\begin{equation}\label{eq:reg}
\forall x, y \in \partial_{\text{in}} R \quad \text{s.t. } \|x - y\| \geq L: \quad
|G_E[H_R](x, y)| \leq e^{-m(L + L^b)}~.
\end{equation}
Otherwise, $R$ is called $E$-singular.

A set $B \subset \mathbb Z^d$ is called $(E, L)$-resonant if there exists $s \in \mathcal B_2 \cap 2^B$ such that $\|G_E[H_s]\| > \exp(\frac{mL^b}{16 J})$; otherwise, $B$ is called $(E, L)$-nonresonant.
\end{df}

\begin{df} Let $J \geq 1$. A collection $\mathfrak S \subset 2^{\mathbb Z^d} \setminus \{\varnothing\}$ is said to be $J$-sparse in $B \subset \mathbb Z^d$ if $\mathfrak S \cap 2^B$ does not contain $J$ pairwise disjoint sets. We colloquially write, for example, ``$E$-resonant $L$-rectangles are $2$-sparse in $s$'' as a shorthand for ``the collection of all $E$-resonant $L$-rectangles is $2$-sparse in the set $s$''.
\end{df}

\begin{prop}\label{p:msa}
For any $m >0$, $b \in (0,1)$, $\gamma \in (2-b, \infty)$ and $J \geq 1$ there exists $L_* = L_*(m,b,\gamma,J,d)$ such that the following holds whenever $L_0 \geq L_*$. Assume that for any $E \subset \mathbb R$
\begin{enumerate}
\item[(1)] for any $k \geq 0$, $(E, L_k)$-resonant $L_{k+1}$-rectangles are $J$-sparse in any $L_{k+2}$-rectangle, and $2$-sparse in the box $[-L_{k+2}, L_{k+2}]^d$;
\item[(2)] $E$-singular $L_0$-rectangles are $J$-sparse in any $L_1$ rectangle.
\end{enumerate}
Then
\begin{enumerate}
\item[(a)] the spectrum of $H$ is pure point;
\item[(b)] for any eigenfunction $\psi$,  $\sup_x |\psi(x)| \exp(\frac{m}{16} \|x\|) < \infty$.
\end{enumerate}
\end{prop}

\begin{rmk} The denominator $16$ in (b) can be replaced with any number greater than $1$.
\end{rmk}

In this section we prove Proposition~\ref{p:msa}, which  will be derived from

\begin{prop}\label{p:ind.reg}
For any $m > 0$, $b \in (0,1)$ and $J \geq 1$ the following holds for $L \geq L_*(m,b,J,d)$. Fix $E \in \mathbb R$, and suppose $R'$ is an $L'$-rectangle such that
\begin{enumerate}
\item[(1)] $E$-singular $L$-rectangles are $J$-sparse in $R'$;
\item[(2)] $R'$ is $(E, L)$-nonresonant;
\item[(3)] $L \leq L' \leq \exp(\frac{mL^b}{100 d J})$.
\end{enumerate}
Then
\begin{enumerate}
\item[(a)] for any $x, y \in R'$ with $\| x - y \| \geq 4 J L$
\begin{equation}\label{eq:subh1} |G_E[H_{R'}](x, y)| \leq e^{-\frac{m}{2}\|x-y\|}~; \end{equation}
\item[(b)] if $100 J L^{2-b} \leq L' \leq \exp(\frac{mL^b}{100 \nu J})$, then $R'$ is  $E$-regular.
\end{enumerate}
\end{prop}

\begin{proof}[Proof of Proposition~\ref{p:msa}]

First, we fix $E$ and prove by induction that, for any $k \geq 0$, $E$-singular $L_k$-rectangles are $J$-sparse in any $L_{k+1}$-rectangle. By the second assumption, this property holds for $k = 0$. Assume that the property holds for some $k$ and fails for $k+1$. Then there is an $L_{k+2}$-rectangle $R''$ containing $J$ disjoint singular $L_{k+1}$-rectangles $R'_j$, $j=1,\cdots,J$. By the induction hypothesis, $E$-singular $L_k$-rectangles are $J$-sparse in each of the $R'_j$. By the first assumption, at least one of them, say, $R_1'$, is $(E, L_k)$-nonresonant. Also, if $L_0$ is large enough, then $L = L_k$ and $L' = L_{k+1} = \lfloor L^\gamma \rfloor$ satisfy the inequalities
\[ 100J L^{2-b} \leq L' \leq \exp(\frac{mL^b}{100 d J})~. \]
Thus $R'_1$ satisfies all the conditions of  part (b) of Proposition~\ref{p:ind.reg}, and is therefore $E$-regular, in contradiction to our assumption.

Second, we show that for any $E$ and $k \geq 0$, and any $(E, L_k)$-nonresonant $L_{k+1}$ rectangle $R'$,
\begin{equation}\label{eq:step2}\forall x,y \in R': \quad \left( \|x-y\| \geq 4 JL_k \,\, \Longrightarrow \,\,
|G_E[H_{R'}](x, y)| \leq \exp(-\frac{m}{2} \|x-y\|)\right)~.\end{equation}
This follows from part (a) of Proposition~\ref{p:ind.reg}, using the first step of the current proof to verify the first condition of the proposition.

Now we are in position to prove the proposition. Schnol's lemma \cite{Ber} implies that for almost any $E$ with respect to the spectral measure of $H$ there exists a non-trivial formal solution $\psi$ of the eigenfunction equation $H \psi = E \psi$ such that $|\psi(x)| \leq (\|x\|+1)^d$. By the first assumption, $(E, L_k)$-resonant $L_{k+1}$-rectangles are $2$-sparse in the box $[-L_{k+2}^d, L_{k+2}^d]$.
By the second step of the current proof, any $(E,L_k)$-nonresonant $L_{k+1}$-rectangle $R'$ satisfies (\ref{eq:step2}),
hence for any point $x \in R'$ with $\operatorname{dist}(x, \partial_{\text{in}}R') \geq 4 JL_k$
\begin{equation}\label{eq:subhbd}\begin{split}
 |\psi(x)| &\leq \sum_{uu' \in \partial R'} |G_E[H_{R'}](x, u)| |\psi(u')| \\
&\leq (3L_{k+1})^d e^{-2mJL_k} (1 + L_{k+2})^d \leq e^{-mJL_k}~. \end{split}\end{equation}
The right-hand side of (\ref{eq:subhbd}) tends to zero as $k \to \infty$. Fix a point $x_*$ such that $\psi(x_*) \neq 0$, then for $k \geq k_0 = k_0(x_*)$ the inequality has to fail, i.e.\ every $L_{k+1}$-rectangle $R' \ni x_*$ such that $\operatorname{dist}(x_*, \partial_{\text{in}} R') \geq 4 J L_k$ has to be $(E, L_k)$-resonant.\footnote{We may assume that for all $k$ $L_{k+1} \geq (10 J)^{100}  L_k$.}

Let $\tilde R' \subset [-L_{k+2}, L_{k+2}]^d \setminus [x_* - 4 JL_k, x_* + 4 J L_k]^d$ be an $L_{k+1}$-rectangle. Then there exists an $L_{k+1}$-rectangle $R'$ disjoint from $\tilde R'$
such that $R' \ni x_*$ and $\operatorname{dist}(x_*, \partial_{\text{in}} R') \geq 4J L_k$. As $R'$ is $(E, L_k)$-resonant, we conclude that $\tilde R'$ is $(E, L_k)$-nonresonant.
This implies that
\begin{equation}\label{eq:bd.shells}\forall k \geq k_0(x_*) \,\, \forall x\,\, \left( \|x\| \in [8JL_k, L_{k+2} - 3 L_{k+1}] \,\, \Longrightarrow \,\,  |\psi(x)| \leq   e^{-mJL_k}\right)~.\end{equation}
In particular, $\psi$ lies in $\ell_2(\mathbb Z^d)$. This holds for every $\psi$, hence the spectrum of $H$ is pure point.

Consider the function $\phi(x) = |\psi(x)|e^{\frac{m}{16}\|x\|}$. From (\ref{eq:bd.shells}), $\phi$ is bounded by $1$ on the set
\[ \bigcup_{k \geq k_0} \left\{x \in \mathbb Z^d \, \mid \, \|x\|\in[8JL_k, 16JL_k]\right\}~.\]
Applying the first inequality in (\ref{eq:subhbd}), we obtain that $\phi$ is bounded by 1 on $\{ \|x\| \geq 8J L_{k_0}\}$. Thus $\phi$ is bounded, as claimed.
\end{proof}

The proof of Proposition~\ref{p:ind.reg} relies on two lemmata. The first one asserts that the Green function $G_E[H_R]$ in (\ref{eq:reg}) can be replaced with $G_E[H_S]$ for $S \supset R$, as long as  $x$ is not very close to the boundary of $R$ in $S$ (in particular, it is required that $x \in \partial_{\text{in}} R \cap \partial_{\text{in}} S$). The following definition will be convenient:

\begin{df}
Let $B$ be a box. An $L$-strip $S \subset B$ is a product $S = I_1' \times \cdots \times I_d'$ of intervals, where $I_j' = I_j$ for $j \neq j_0$, and $\# I_{j_0}' = L$. A set is called a strip if it is an $L$-strip for some value of $L$.
\end{df}

\begin{lemma}\label{l:1}
In the setting of Proposition~\ref{p:ind.reg}, let $R \subset R'$ be an $E$-regular $L$-rectangle, and let $R \subset S \subset R'$ be a strip (see Figure~\ref{fig:lemma1}). Then
\begin{equation}
\forall x, y \in \partial_{\text{in}} R  \text{ s.t.} \operatorname{dist}(x, \{y\} \cup (S \setminus R)) \geq L: \,
|G_E[H_S](x, y)| \leq e^{-m(L + \frac12 L^b)}~.
\end{equation}
\begin{figure}
\begin{center}
\begin{tikzpicture}
    \clip (-1,-.7) rectangle (10cm,4.1cm);

    \fill [color=gray!30!white ] (-.1,-.1) rectangle (9.1,1.1) ;
    \draw [draw=red,line width=0.3mm] (-.12,-.12) rectangle (9.12,1.12) ;
	\node[text=red] at (-.5,.5) {$S$};

	\draw [draw=blue,line width=0.3mm] (4.78,-.22) rectangle (7.32,1.22) ;
	\fill [color=gray!80!white] (4.8,-.2) rectangle (7.3,1.2) ;
	\node[text=blue] at (6,-.5) {$R$};
	\node[text=purple] at (5.5,-.5) {$x$};
	\draw[style=help lines,line width=0.3mm] (0,0) grid[step=.5cm] (9,4);
	\node[text=black] at (-.5,2) {$R'$};
	
    \foreach \x in {0,1,...,18}
    {
        \foreach \y in {0,1,...,8}
        {
            \node[draw,circle,inner sep=2pt,fill] at (.5*\x,.5*\y) {};
        }
    }
    \draw [draw=black!80!white,line width=0.3mm] (4.8,-.2) rectangle (7.3,1.2) ;
	\draw[draw=purple, line width=0.3mm]  (5.5,0) circle (7pt);
\end{tikzpicture}
\end{center}
\caption{Illustration to Lemma~\ref{l:1}. In this case $d = 2$, $L = 2$ and $L' = 8$; $y$ can be any vertex on $\partial_{\text{in}} R$ except for $x$ and the two vertices adjacent to it.}\label{fig:lemma1}
\end{figure}
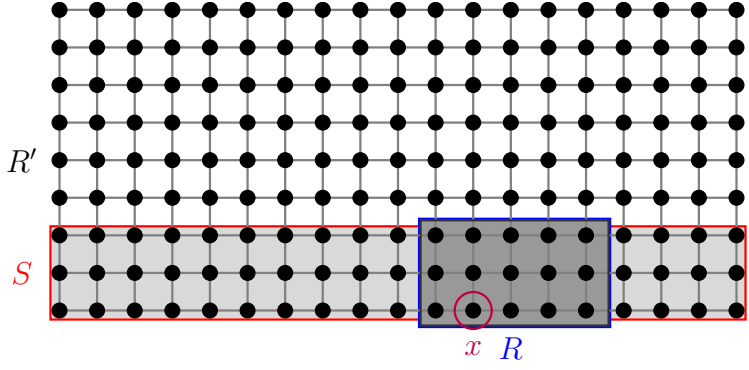
\end{lemma}

\begin{proof}
By assumption (2), the rectangle $R'$ is $(E,L)$-nonresonant, hence by the resolvent identity
\[ \begin{split}
|G_E[H_S](x, y)|
&\leq |G_E[H_R](x, y)| + \sum_{uu' \in \partial R \setminus \partial S} |G_E[H_R](x, u)| |G_E[H_S](u', y)| \\
&\leq \exp(-m (L + L^b)) \left[ 1 + (CL)^{d - 1} \exp(\frac{mL^b}{16 J}) \right] \\
&\leq \exp(-m(L + \frac12 L^b))
\end{split}\]
if $L$ is sufficiently large, $L \geq L_*(m,b,J,d)$.
\end{proof}

\begin{lemma}\label{l:2} In the setting of Proposition~\ref{p:ind.reg}, suppose $B \subset R'$ is a box. Let $x, y \in \partial_{\text{in}} B$, and let $S \subset B$ be an $L$-strip such that $x \in \partial_{\text{in}} S$ and $y \notin S$. Construct an $L$-rectangle $R \subset S$ as in Figure~\ref{fig:lemma2}, left, so that $x$ is the centre of a large face of $R$ (if $x$ is close to the boundary of $S$, align $R$ with the boundary, as in Figure~\ref{fig:lemma2}, right).
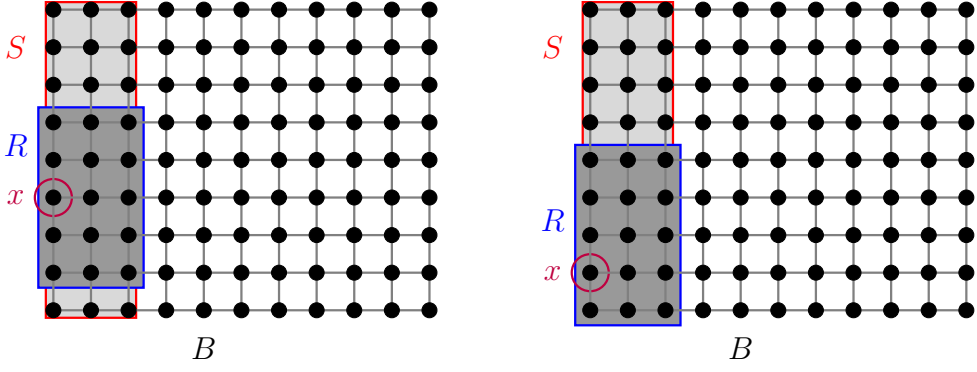
\begin{figure}
\begin{center}
\begin{tikzpicture}
    \clip (-1,-.7) rectangle (6cm,5cm);
    \fill [color=gray!30!white ] (-.1,-.1) rectangle (1.1,4.1) ;
	\draw [draw=red,line width=0.3mm] (-.1,-.1) rectangle (1.1,4.1) ;
	\node[text=red] at (-.5,3.5) {$S$};
	\fill [color=gray!80!white]  (-.2,.3) rectangle (1.2,2.7);
	\draw [draw=blue,line width=0.3mm] (-.2,.3) rectangle (1.2,2.7) ;
	\node[text=blue] at (-.5,2.2) {$R$};
	\draw[draw=purple, line width=0.3mm]  (0,1.5) circle (7pt);
	\node[text=purple] at (-.5,1.5) {$x$};
	\draw[style=help lines,line width=0.3mm] (0,0) grid[step=.5cm] (5,4);
	\node[text=black] at (2,-.5) {$B$};
	
    \foreach \x in {0,1,...,10}
    {
        \foreach \y in {0,1,...,8}
        {
            \node[draw,circle,inner sep=2pt,fill] at (.5*\x,.5*\y) {};
        }
    }
\end{tikzpicture}
\begin{tikzpicture}
    \clip (-1,-.7) rectangle (6cm,5cm);
     \fill [color=gray!30!white ] (-.1,-.1) rectangle (1.1,4.1) ;
	\draw [draw=red,line width=0.3mm] (-.1,-.1) rectangle (1.1,4.1) ;
	\node[text=red] at (-.5,3.5) {$S$};
	\fill [color=gray!80!white]  (-.2,-.2) rectangle (1.2,2.2) ;
	\draw [draw=blue,line width=0.3mm] (-.2,-.2) rectangle (1.2,2.2) ;
	\node[text=blue] at (-.5,1.2) {$R$};
	\draw[draw=purple, line width=0.3mm]  (0,.5) circle (7pt);
	\node[text=purple] at (-.5,.5) {$x$};
	\draw[style=help lines,line width=0.3mm] (0,0) grid[step=.5cm] (5,4);
	\node[text=black] at (2,-.5) {$B$};
	
    \foreach \x in {0,1,...,10}
    {
        \foreach \y in {0,1,...,8}
        {
            \node[draw,circle,inner sep=2pt,fill] at (.5*\x,.5*\y) {};
        }
    }
\end{tikzpicture}
\end{center}
\caption{Illustration to Lemma~\ref{l:2}: $d = 2$, $L = 3$. Note that the strip $S$ could also be horizontal.}
\label{fig:lemma2}
\end{figure}
Then
\begin{enumerate}
\item if $R$ is regular, then
\[ |G_E[H_B](x, y)| \leq e^{-m(L+ \frac13 L^b)} \max_{vv' \in \partial S \setminus \partial B} |G_E[H_{B \setminus S}](v', y)|~;\]
\item if $R$ is singular, then
\[ |G_E[H_B](x, y)| \leq e^{+ \frac{mL^b}{8J}} \max_{vv' \in \partial S \setminus \partial B} |G_E[H_{B \setminus S}](v', y)|~. \]
\end{enumerate}
\end{lemma}

\begin{proof} If $R$ is regular, by the resolvent identity,
\[\begin{split} |G_E[H_B](x, y)|
&\leq \sum_{uu' \in \partial R \setminus \partial B} |G_E[H_B](x, u)| |G_E[H_{B \setminus R}] (u', y)| \\
&\leq  \sum_{uu' \in \partial R \setminus \partial B} \sum_{vv' \in \partial S \setminus \partial B} |G_E[H_B](x, u)| |G_E[H_{B \setminus R}](u', v)| |G_E[H_{B \setminus S}](v', y)|~.
\end{split}\]
According to Lemma~\ref{l:1}, $|G_E[H_B](x, u)| \leq e^{-m (L + \frac12 L^b)}$, hence
\[\begin{split}
|G_E[H_B](x, y)|
&\leq (2L)^{\nu - 1} (2L')^\nu e^{-m(L+\frac12L^b)} e^{\frac{mL^b}{8J}} \max_{vv' \in \partial S \setminus \partial B} |G_E[H_{B\setminus S}](v', y)| \\
&\leq e^{-m(L+\frac13 L^b)} \max_{vv' \in \partial S \setminus \partial B} |G_E[H_{B\setminus S}](v', y)|~.
\end{split}\]
If $R$ is singular, we argue similarly, starting from the estimate
\[ |G_E[H_B](x, y)| \leq \sum_{vv' \in \partial S \setminus \partial B} |G_E[H_B](x, v)| |G_E[H_{B\setminus S}](v', y)|~. \]
\end{proof}

\begin{proof}[Proof of Proposition~\ref{p:ind.reg}]
Suppose $x, y \in \partial_{\text{in}} R'$, $\|x - y\|\geq L'$. Iterating Lemma~\ref{l:2}, we obtain
\begin{equation}\label{eq:1}
\begin{split}
|G_E[H_{R'}](x, y)|
&\leq e^{\frac{mL^b}{16 J}} e^{-m(L + \frac13 L^b)(\frac{L'}{L} - J)} e^{\frac{mL^b}{8J}} \\
&\leq \exp \left[ m \left\{ - L' + L^b \left(\frac{1}{5J} + \frac13 J \right) - \frac13 L' L^{b-1} + JL\right\} \right] \\
&\leq \exp \left[ m (- L' - \frac13 L' L^{b-1} + 2 JL)\right]~.
\end{split}\end{equation}
If $L' \geq 100 J L^{2-b}$, then
\[ \frac13 L^{b-1}L' \geq 2J L + L'^{b}~, \]
hence
\[ (\ref{eq:1}) \leq \exp(-m(L' + L'^b))~. \]
For arbitrary $L'$ and $x, y \in R'$ with $\|x - y\| \geq 4 J L$, a similar argument yields
\[ |G_E[H_{R'}](x, y)| \leq e^{-\frac{m}{2}\|x-y\|}~. \]
\end{proof}

\subsection{Wegner estimate, and Proof of Proposition~\ref{t:1}}\label{s:pft1}

Let $H(\omega, \theta;g)$ be an operator of the form
\begin{equation}\label{eq:schr'} (H(\omega, \theta; g) f)(x) = \sum_{\|y-x\|=1} f(y) + g v(T^x \omega, \theta) f(x)~.\end{equation}
We recall our basic assumptions:
\begin{align}
&\text{\bf (UPA)}_A&\quad
&\inf_\omega \min_{0 < \|x\| \leq L} \operatorname{dist}(T^x \omega, \omega) \geq c L^{-A}\\
&\text{\bf (LIB)}_\eta&\quad
&p_\omega (t \, \mid \, \Omega \setminus Q_\epsilon(\omega)) \leq \exp(C \epsilon^{-\eta})~, \quad \epsilon \in (0,1/2]\\
&\text{\bf (NET)}_\nu&\quad
&\min \# (\text{$\epsilon$-net in $\Omega$}) \leq (C/\epsilon)^\nu~, \quad \epsilon \in (0, 1] \\
&\text{\bf (UH\"ol)}_\kappa& \quad &\lim_{R\to\infty} \mathbb P^\Theta(\mathfrak H_R) = 1~,
\end{align}
where $\mathfrak H_R$ is the collection of $\theta \in \Theta$ such that $\|v(\cdot, \theta)\|_\infty \leq R$ and $v(\cdot, \theta)$ is uniformly $\kappa$-H\"older with constant $R$:
\begin{equation}\label{eq:R}\sup_{\omega} |v(\omega, \theta)| + \sup_{\omega' \neq \omega} \frac{|v(\omega', \theta) - v(\omega,\theta)|}{\operatorname{dist}(\omega', \omega)^\kappa} \leq R~.\end{equation}

\begin{prop}\label{p:weg}
Assume that $\text{\bf (UPA)}_A$, $\text{\bf (LIB)}_\eta$, $\text{\bf (NET)}_\nu$ and $\text{\bf (UH\"ol)}_\kappa$ hold with $A\eta < 1$. Let
\[ m = 16~, \quad J = \min ( \mathbb Z \cap (\frac\nu\kappa+1, \infty))~, \]
and choose $b \in (0, 1)$ and $\gamma \in (2-b, \infty)$ so that $A\eta < b/\gamma^2$. Then there exist two measurable  functions $L_{\min}(\omega, \theta)$ and $g_{\min}(\omega, \theta)$ that are $\Theta$-almost-everywhere finite for each $\omega \in \Omega$, such that for $L_0 \geq L_{\min}$, $g \geq g_{\min}$ the assumptions (1)--(2) of Proposition~\ref{p:msa} hold for the operator $H(\omega, \theta; g)$.
\end{prop}

The proof is based on the following lemma. For $r > 0$, $E \in \mathbb R$, $\omega \in \Omega$ and  $s_1, \cdots, s_k \subset \mathbb Z^d$, define the following  events in $\Theta$:
\begin{align}
&\operatorname{Reson}_{L,r}(s_1, \cdots, s_k; \omega; E) &=&
\left\{ \forall j=1,\cdots,k \, \|G_E[H_{s_j}(\omega,\theta; g)]\| > \frac{e^{L^r}}{g}\right\}\\
&\operatorname{Reson}_{L,r}(s_1, \cdots, s_k; \omega) &=&
\bigcup_{E \in \mathbb R} \operatorname{Reson}_{L,r}(s_1, \cdots, s_k; \omega; E)  \\
&\operatorname{Reson}_{L,r}(s_1, \cdots, s_k) &=&
\bigcup_{\omega \in \Omega} \operatorname{Reson}_{L,r}(s_1, \cdots, s_k; \omega)
\end{align}

\begin{lemma}\label{l:weg} Assume that $\text{\bf (UPA)}_A$, $\text{\bf (LIB)}_\eta$, $\text{\bf (NET)}_\nu$  hold with $A\eta < 1$. Let $m,b,\gamma,J$ be as in Proposition~\ref{p:weg}, and let  $r > A\eta$, $R \geq 1$.\footnote{Eventually, $r$ will be taken to be slightly greater than $A\eta$, however, no upper bound is formally required in the current lemma. $R$ will eventually play the same r\^ole as in (\ref{eq:R}).}  Then
\begin{enumerate}
\item for $k \geq 2$,
\[ \sup_{\omega \in \Omega}\sup_{s_1,\cdots,s_k} \mathbb P^\Theta (\operatorname{Reson}_{L,r}(s_1,\cdots,s_k; \omega) \cap \mathfrak H_R) \leq  R \exp(-(k-1)L^r - o(L^r))~;\]
\item for $k > \frac{\nu}{\kappa}+1$,
\[\sup_{s_1,\cdots,s_k} \mathbb P^\Theta (\operatorname{Reson}_{L,r}(s_1,\cdots,s_k) \cap \mathfrak{H}_R) \leq  R^{\frac\nu\kappa+1} \exp\left(-\big(k-\frac\nu\kappa-1\big)L^r - o(L^r)\right)~,\]
\end{enumerate}
where the supremum in the first formula and the interior one in the second formula are over $k$-tuples of pairwise disjoint subsets of $[-L,L]^d$.
\end{lemma}

\begin{proof} Fix $\omega \in \Omega$ and $E \in \mathbb R$. From $\text{\bf (UPA)}_A$ and $\text{\bf (LIB)}_\eta\,$, the joint probability density (in $\Theta$) of $(V(x; \omega))_{x \in B}$, $B \subset [-L,L]^d$, is bounded by
\[ \left(\frac{\exp(C(cL^{-A})^{-\eta})}{g}\right)^{\# B}~,\]
therefore by the usual Wegner argument \cite{W81,AWbook}, we obtain  that for $M > 0$
\begin{equation}\label{eq:est1} \begin{split}
&\mathbb P^\Theta \left\{ \forall j=1,\cdots,k \,\, \| G_E[H_{s_j}(\omega,\theta)]\| >  M \right\} \\
&\quad\leq \left( \frac{\exp(C(cL^{-A})^{-\eta})}{gM}\right)^k \prod_{j=1}^k \# s_j
\leq \left( \frac{(3L)^d \exp(C_1 L^{A\eta})}{gM} \right)^k~.
\end{split}
\end{equation}
Let $M = \frac{1}{4g} \exp(L^r)$; then
\[ \text{RHS of (\ref{eq:est1})}  \leq \left[ 4 (3L)^d \exp(C_1 L^{A\eta}- L^r)\right]^k
\leq \exp(-k L^r + o(L^r))~;\]
here and in the sequel the implicit constants are uniform in $s_j$ and $\omega$.
Let $\mathcal N_\Omega$ be an $(4gMR)^{-1/\kappa}$-net in $\Omega$, and $\mathcal N_{\mathbb R}$ -- a $(4M)^{-1}$-net in $[-10dgR, 10dgR]$, chosen so that
\[ \# \mathcal N_\Omega \leq (CgMR)^{\nu/\kappa}~, \quad
\# \mathcal N_{\mathbb R} \leq Cdg M R~. \]
Then
\begin{equation}\label{eq:fixedomega}\begin{split}
&\mathbb P^\Theta \left\{ \exists E \in \mathcal N_R: \,\, \forall j=1,\cdots,k \,\, \| G_E[H_{s_j}(\omega,\theta)]\| \geq M \right\} \\
&\quad\leq CdgMR \exp(-kL^r + o(L^r)) \leq R \exp(-(k-1)L^r + o(L^r))
\end{split}\end{equation}
for any $\omega \in \Omega$, and
\begin{equation}\label{eq:allomega}\begin{split}
&\mathbb P^\Theta \left\{ \exists E \in \mathcal N_R, \, \omega \in \mathcal N_\Omega: \,\, \forall j=1,\cdots,k \,\, \| G_E[H_{s_j}(\omega,\theta)]\| \geq M \right\} \\
&\quad\leq (CgMR)^{\frac{\nu}{\kappa}} R\exp(-(k-1)L^r + o(L^r)) \\
&\quad\leq R^{\frac\nu\kappa+1} \exp(-(k-\frac\nu\kappa-1)L^r + o(L^r))~.
\end{split}\end{equation}
If $\|G_E[H_s(\omega, \theta)]\| \leq M$, $\theta \in \mathfrak H_R$, $|E'-E| \leq \frac{1}{4M}$, and $\operatorname{dist}(\omega', \omega) \leq (4gMR)^{-1/\nu}$, then
\begin{equation}\label{eq:bdgr} \|G_{E'}[H_s(\omega', \theta)]\| \leq 2M~.\end{equation}
Also note that on $\mathfrak H_R$ the bound (\ref{eq:bdgr}) holds for all $|E| \geq 10d g R$: indeed, such energies are at distance $\geq 1$ from the spectrum of $H$,
Therefore (\ref{eq:fixedomega}) and (\ref{eq:allomega}) imply the first and second assertions of the lemma, respectively.
\end{proof}

\begin{proof}[Proof of Proposition~\ref{p:weg}]
Fix $\omega_0 \in \Omega$. Denote by $\operatorname{Bad}_L(\omega_0)$ the event (in $\Theta$-space) that either there exist $E \in \mathbb R$ and $\omega \in \Omega$ such that $(E, L)$-resonant $\lfloor L^\gamma\rfloor$-rectangles are not $J$-sparse in
\[ B_L = [-\lfloor \lfloor L^\gamma\rfloor^\gamma\rfloor, \lfloor \lfloor L^\gamma\rfloor^\gamma\rfloor]^d~, \]
 for $H(\omega, \theta)$, or there exists $E$ such that $(E, L)$-resonant $\lfloor L^\gamma\rfloor$-rectangles are not $2$-sparse in $B_L$
for $H(\omega_0, \theta)$. According to Lemma~\ref{l:weg} applied with an arbitrary $r \in (A\eta, b/\gamma^2)$ and with $\lfloor \lfloor L^\gamma\rfloor^\gamma\rfloor$ in place of $L$,
\[ \mathbb{P}(\operatorname{Bad_L} \cap \mathfrak H_R)
\leq R^{\frac{\nu}{\kappa}+1} \exp(- c L^r + o(L^r))~,\]
where
$c= \min(J-\frac\nu\kappa-1, 1) > 0$.
Thus for every $R \geq 1$
\[ \mathbb P( \limsup_{L \to \infty}  \operatorname{Bad}_L \cap \mathfrak H_R) = 0~.\]
Combining this with $\text{\bf (UH\"ol)}_\kappa$, we obtain that almost every $\theta$ lies in $\mathfrak H_R \setminus \operatorname{Bad}_L$ for all sufficiently large $R$ and $L$ (i.e.\ $R \geq R_{\min}(\theta)$ and $L \geq L_{\min}(\theta)$).

Then for $L_0 \geq L_{\min}(\theta)$ each $H(\omega, \theta)$ satisfies that for all $k \geq 0$ $(E, L_k)$-resonant $L_{k+1}$-rectangles are $J$-sparse in any $L_{k+2}$-rectangle. Indeed, the restriction of  $H(\omega, \theta)$ to any $L_{k+2}$-rectangle coincides with the restriction of $H(\omega', \theta)$ to $[-L_{k+2}, L_{k+2}]^{d-1}\times[1,L_{k+2}]$ for an appropriately chosen $\omega'$. Also, for $H(\omega_0,\theta)$, $(E, L_k)$-resonant $L_{k+1}$-rectangles  and $2$-sparse in $[-L_{k+2},L_{k+2}]^d$. Thus the first half of assumption (1) of Proposition~\ref{p:msa} holds.

Next, let $g\geq  10^{10} d  e^{L^r}$. For any $L_1$-rectangle $R'$ and any disjoint $L_0$-rectangles $R_1, \cdots, R_J \subset R'$, there exists $j \in \{1, \cdots, J\}$ such that
\[ \|G_E[H_{R_j}]\| \leq  \frac{\exp(L^r)}{g}~, \quad \text{i.e.} \quad \operatorname{dist}(E, \sigma(H_{R_j}))\geq  \frac{g}{ \exp(L^r)} \geq 10^{10} d~, \]
therefore $R_j$ is $E$-regular by the Combes--Thomas bound \cite{AWbook}. Hence also asumption (2) of Proposition~\ref{p:msa} holds.
\end{proof}

\begin{proof}[Proof of Proposition~\ref{t:1}]
For every $\omega$ and almost every  $\theta$ there exist $L_{\min}$ and $g_{\min}$ such that the assumptions of Proposition~\ref{p:msa} hold for $L \geq L_{\min}$ and $g \geq g_{\min}$. Denote by $\operatorname{Assum}_{g,L}$ the set of $(\omega,\theta)$ for which these assumptions hold with the given values $g$ and $L$. Then for any $\delta > 0$ there exist $L_\delta$ and $g_\delta$ such that for $L \geq L_\delta$ and $g \geq g_\delta$
\[ \mathbb P_{\Omega \times \Theta} (\operatorname{Assum}_{g,L}) \geq 1 - \delta~. \]
Denote
\[ \operatorname{Assum}_{g,L}^\theta = \left\{ \omega: (\omega,\theta) \in \operatorname{Assum}_{g,L} \right\}~. \]
Then
\[ \mathbb{P}_\Theta\left( \left\{ \theta: \, P_\Omega(\operatorname{Assum}_{g,L}^\theta) \leq \frac12 \right\} \right)\leq 2\delta~. \]
If $\theta$ does not lie in this set, then by ergodicity there exists a shift of the operator $H(\omega, \theta)$ for which the  the assumptions of Proposition~\ref{p:msa} hold. Invoking Proposition~\ref{p:msa}, we obtain the result.
\end{proof}

\section{Interpolation of Gaussian processes}\label{S:int}

The general strategy is as follows. A lemma of \cite{Kar}, which we reproduce in Section~\ref{s:kar}, reduces the proof of Proposition~\ref{p:var} to the construction of a compactly supported function with prescribed decay of the Fourier transform. In Section~\ref{s:decay} we construct such a function by adjusting the arguments of  \cite{PW,Lev,Ron}.

\subsection{A formula of Karhunen}\label{s:kar}
We use the conventions
\begin{align}
&\hat g(\lambda) = \int g(\xi) \exp(- i \langle \xi, \lambda \rangle) d\xi\\
&\check h(\xi) = \int h(\lambda) \exp( i \langle \xi, \lambda \rangle) \frac{d\lambda}{(2\pi)^\nu}
\end{align}
for the Fourier transform of $g: \mathbb R^\nu \to \mathbb C$ and its inverse, and
\begin{align}
&\hat g(\ell) = \int_{\mathbb T^\nu} g(\omega) \exp(- i \langle \omega, \ell \rangle) d\xi\\
&\check h(\omega) = \sum_{\ell \in 2\pi \mathbb Z^\nu} h(\ell) \exp( i \langle \omega, \ell \rangle)
\end{align}
for the Fourier transform of $g:  \mathbb T^\nu \to \mathbb C$ and its inverse. With these conventions,
\begin{align}
&\int_{\mathbb R^\nu} |\hat g(\lambda)|^2 d\lambda = (2\pi)^\nu \int_{\mathbb R^\nu} |g(\xi)|^2 d\xi& \quad &(\mathbb R^\nu)&\\
&\sum_{\ell \in 2\pi \mathbb Z^\nu} |\hat g(\ell)|^2 =   \int_{\TT^\nu} |g(\xi)|^2 d\xi& \quad &(\mathbb T^\nu)~.& \end{align}

The following lemma goes back to the work of  \cite{Kar} (see further \cite[\S 4.13, Test 2]{DM}).

\begin{lemma}[Karhunen]\label{l:kar} For $v(\omega)$ as in (\ref{eq:strg}),
\[ \mathbf V(\epsilon) \overset{\text{def}}{=} \operatorname{Var} \left( v(\omega) \, \big| \, \{ v(\omega') \, : \, \|\omega' - \omega\| \geq \epsilon \}\right)  = \sup \left\{ \frac{|g(0)|^2}{\sum_{\ell} |\hat g(\ell)|^2 W(\ell)} \, \big| \, \operatorname{supp} g \subset \{ \|\omega\|<\epsilon \} \right\}~.\]
\end{lemma}

\begin{proof} We prove the inequality ``$\geq$", as this is the direction we use in the sequel. Let $\tilde v$ be an independent copy of $v$, and let
\[ X(\omega) =  \frac{v(\omega) + \tilde v(\omega)}{\sqrt{2}} = \sum_{\ell \in 2\pi \mathbb Z^\nu} \frac{G_\ell e^{i \langle \omega, \ell \rangle}}{\sqrt{W(\ell)}}~,\]
where $G_\ell$ are independent standard {\em complex} Gaussian variables. It suffices to prove the equality for $\mathbf V(\epsilon)$ defined for $X$ in place of $v$.
We start from the relation
\[ \mathbf V(\epsilon) = \inf \left\{ \mathbb E \left| X(0) - \int X(\omega) \rho(\omega) d\omega \right|^2 \, \big| \, \rho \in L_2(\TT^\nu)~, \,\, \operatorname{supp} \rho \subset \{ \|\xi \|\geq \epsilon\} \right\}~. \]
Rewrite
\[ \begin{split}
&\mathbb E \left| X(0) - \int X(\omega) \rho(\omega) d\omega \right|^2\\
&\qquad= \mathbb E \left| \sum_{\ell \in 2\pi \mathbb Z^d} \frac{G_\ell}{\sqrt{W(\ell)}} \left(1 - \int e^{i \langle \omega, \ell \rangle} \rho(\omega) d\omega \right)\right|^2\\
&\qquad = \mathbb E \left| \sum_{\ell \in 2\pi \mathbb Z^d} \frac{G_\ell}{\sqrt{W(\ell)}} (1 - \overline{\hat{\rho}(\ell)}) \right|^2 = \sum_{\ell \in 2\pi \mathbb Z^d} \frac{| 1 - \overline{\hat{\rho}(\ell)})|^2}{W(\ell)}~.
\end{split}\]
For an arbitrary $\rho$ supported in $\{ \|\omega \|\geq \epsilon \}$ and an arbitrary $g$ supported in $\{\|\omega\|\leq \epsilon\}$,
\[ g(0) = g(0) - \int g(\omega) \rho(\omega) d\omega = \sum \hat g(\ell) (1 - \overline{\hat \rho(\ell)})~,\]
whence by Cauchy--Schwarz
\[ |g(0)|^2 \leq \left(\sum |\hat g(\ell)|^2 W(\ell)\right) \times \left(\sum \frac{|1 - \overline{\hat \rho(\ell)}|^2}{W(\ell)}\right)~. \]
Thus
\[ \mathbf V(\epsilon) \geq \frac{|g(0)|^2}{\sum_{\ell} |\hat g(\ell)|^2 W(\ell)}~. \qedhere\]
\end{proof}

\subsection{Functions with prescribed Fourier decay}\label{s:decay}

The following proposition is a quantitative version of a result proved in \cite{PW} and  \cite{Lev} in dimension $\nu = 1$, and in \cite{Ron} in arbitrary dimension. The method of convolutions used in the proof was applied
for similar purpose already in \cite{Lev}, and for the proof of necessity in the Denjoy--Carleman
theorem -- in \cite{Man} (where an earlier unpublished work of Bray is quoted) and in \cite{Bang}; see further  \cite[\S 1.3 and Notes]{Hor}  and  \cite[\S 25]{Levin}.
\begin{prop}\label{prop:fdec} Let $M: \mathbb R_+ \to \mathbb R_+$ be a nondecreasing function such that
\[ M(0) = 1~, \quad \int^\infty \frac{\log M(t)}{t^2} dt < \infty~. \]
Then for any $\nu \geq 1$ and $\epsilon \in  (0, 1]$ there exists $g: \mathbb R^\nu \to \mathbb R_+$ such that
\begin{align}
& \operatorname{supp} g \in [-\epsilon, \epsilon]^\nu~,&  & g(0) =
\max g~,&  &\hat g(0) = 1~,\label{eq:gprop1} \\
& |\hat g(\lambda)| \leq \frac{e M(S^{-1}(\epsilon/e))}{M(\|\lambda\|)}~,&   &\text{where} \quad S(t) = \int_t^\infty \frac{\log M(\tau)}{\tau^2} d\tau~. & \label{eq:gprop2}
\end{align}
\end{prop}

\begin{proof}
Let $u(\xi) = 2^{-\nu} \mathbbm{1}_{[-1,1]^\nu}(\xi)$, so that $\hat u(\lambda) = \prod_{r = 1}^\nu \frac{\sin \lambda_r}{\lambda_r}$. Then
\begin{equation}\label{eq:bdu} |\hat u(\lambda)| \leq \min(1, \|\lambda\|^{-1})~.\end{equation}
We may assume that $M$ is continuous. Let
\[ R_j = \min \left\{ t \geq 0\, \mid \, M(t) = e^j\right\}~, \]
and choose $k_0$ so that
\[ S(R_{k_0}) \leq \frac\epsilon{e}~, \, S(R_{k_0-1}) > \frac\epsilon{e}~. \]
Define
\[ \hat g(\lambda) = \prod_{j=k_0}^\infty \hat u(\frac{e\lambda}{R_j})~. \]
Then $\max \hat g = g(0)$ and $\hat g(0)= 1$, and
\[ \operatorname{supp} g \subset [- \sum_{j = k_0}^\infty \frac{e}{R_j}, \sum_{j = k_0}^\infty \frac{e}{R_j}] \subset [-\epsilon, \epsilon]^\nu~, \]
since
\[\begin{split} \sum_{j = k_0}^\infty \frac{1}{R_j}
&= \int_{R_{k_0}}^\infty \frac{dt}{t^2} \# \left\{ k_0 \leq j \leq t \right\}\\
&\leq \sum_{j \geq k_0} \int_{R_j}^{R_{j+1}} \frac{dt}{t^2} (j - k_0 + 1)_+ \\
&\leq \sum_{j \geq k_0} \int_{R_j}^{R_{j+1}}  \frac{\log M(t)}{t^2} dt =S(R_{k_0}) \leq \frac\epsilon{e}~.
\end{split}\]
This proves (\ref{eq:gprop1}), and we turn to the proof of (\ref{eq:gprop2}). By (\ref{eq:bdu}), we have for $R_k \leq \| \lambda \| < R_{k+1}$:
\[\begin{split} |\hat g(\lambda)| &\leq \prod_{j\geq k_0} \min(1, \frac{R_j}{e\|\lambda\|}) \\
&\leq \prod_{j=k_0}^k \frac{1}{e} = \exp(- (k - k_0+1)_+)~.
\end{split}\]
On the other hand,
\[ M(\|\lambda\|) \leq M(R_{k+1}) \leq \exp(k+1)~. \]
Hence
\[ |\hat g(\lambda)| \leq e^{k_0} / M(\|\lambda\|) \leq e M(S^{-1}(\epsilon/e)) / M(\|\lambda\|)~, \]
as claimed.
\end{proof}

\subsection{Proof of Proposition~\ref{p:var}}\label{s:var}

We apply Proposition~\ref{prop:fdec} with $M_1(t) = \sqrt{M(t)}$, and $S_1(t) = \frac{1}{2} S(t)$. The function $g$ thus obtained satisfies
\[ |\hat g(\ell)| \leq \frac{e M_1(S_1^{-1}(\epsilon/e))}{M_1(\|\ell\|)} = \frac{e \sqrt{M(S^{-1}(\frac{2}{e} \epsilon)}}{\sqrt{M(\|\ell\|)}}~,\]
whence
\[ \sum |\hat g(\ell)|^2 W(\ell)
\leq K \max |\hat g(\ell)|^2 M(\ell) \leq e^2 K M(S^{-1}(\frac2e \epsilon))~. \]
On the other hand,
\[ |g(0)|^2 = \max_\omega |g(\omega)|^2 \geq \left[ \frac{1}{(2\epsilon)^\nu} \int g(\omega) d\omega \right]^2 = \frac{1}{(2\epsilon)^{2\nu}}~. \]
Thus by Lemma~\ref{l:kar}
\[ \mathbf V(\epsilon) \geq \frac{1}{e^22^{2\nu}  K \epsilon^{2\nu} M(S^{-1}(\frac2e \epsilon))}~, \]
as claimed.
\qed

\paragraph{Acknowledgements.}
Parts of this work were completed while the authors enjoyed the hospitality of the Isaac Newton Institute, the Weizmann Institute of Science, and the Mittag-Leffler Institute. SS is supported in part by the European Research Council starting  grant
639305 (SPECTRUM) and by a Royal Society Wolfson Research Merit Award.

We are grateful to Olga Izyumtseva for helpful comments, and particularly for bringing the works \cite{Cuzick,CuzickDupreez} to our attention.

\end{document}